\documentclass[10pt]{proc}

\usepackage{latexsym, amsmath, amsfonts, amssymb, amsthm, graphicx, epsfig, array, longtable, calc, multirow, hhline, verbatim, color, hyperref, authblk, enumerate} 
\usepackage[square,sort,comma,numbers]{natbib}
\usepackage{amsmath, amssymb, amsfonts, amsthm, hyperref, verbatim, latexsym, epsfig, array, longtable, multirow, hhline, color}
\usepackage{subfig}
\usepackage{graphicx}

\def\A{\mathbf A}

\def\V{\mathbf V}
\def\tr{\mathrm{tr}}

\def\b{\mathbf b}

\def\e{\mathbf e}
\def\f{\mathbf f}
\def\g{\mathbf g}
\def\h{\mathbf h}

\newcommand\x{\mathbf x}

\def\z{\mathbf z}

\def\F{\mathbf F}

\def\I{\mathbf I}
\def\M{\mathbf M}
\def\N{\mathbf N}

\def\1{\mathbf 1}
\def\0{\mathbf 0}

\newcommand\Pb{\mathbf P}
\newcommand\Rb{\mathbf R}

\newcommand\Cc{\mathcal{C}}
\newcommand\Fc{\mathcal{F}}
\newcommand\Ss{\mathcal{S}}
\newcommand\Rbb{\mathbb R}

\newcommand\vol{\mathrm{vol}}

\newcommand\MVEE{\mathrm{MVEE}}
\newcommand\ind{\textit{ind}}

\newtheorem{theorem}{Theorem}
\newtheorem{lemma}{Lemma}

\theoremstyle{definition}
\newtheorem{example}{Example}

\theoremstyle{plain}

\title{On greedy heuristics for computing\\D-efficient saturated subsets}
\author{Radoslav Harman\thanks{Faculty of Mathematics, Physics and Informatics, Comenius University, Bratislava\\ email: harman@fmph.uniba.sk}, Samuel Rosa\thanks{Faculty of Mathematics, Physics and Informatics, Comenius University, Bratislava, Slovakia}}

\begin{document}
\maketitle

\begin{abstract}
Let $\Fc$ be a set consisting of $n$ real vectors of dimension $m \leq n$. For any saturated, i.e., $m$-element, subset $\Ss$ of $\Fc$, let $\mathrm{vol}(\Ss)$ be the volume of the parallelotope formed by the vectors of $\Ss$. A set $\Ss^*$ is called a $D$-optimal saturated subset of $\Fc$, if it maximizes $\mathrm{vol}(\Ss)$ among all saturated subsets of $\Fc$. In this paper, we propose two greedy heuristics for the construction of saturated subsets performing well with respect to the criterion of $D$-optimality: an improvement of the method suggested by Galil and Kiefer for the initiation of $D$-optimal experimental design algorithms, and a modification of the Kumar-Yildirim method, the original version of which was proposed for the initiation of the minimum-volume enclosing ellipsoid algorithms. We provide geometric and analytic insights into the two methods, and compare them to the commonly used random and regularized greedy heuristics. We also suggest variants of the greedy methods for a large set $\Fc$, for the construction of $D$-efficient non-saturated subsets, and for alternative optimality criteria.

\end{abstract}

\section{Introduction}

Consider a set $\Fc$ consisting of vectors $\f_1,\ldots,\f_n \in \Rbb^m$ which we will call ``regressors'', and let $2 \leq m \leq n$. Our aim is to select an $s$-element subset of $\Fc$, $s \leq n$, focusing on the ``saturated'' subsets, that is, $s = m$. The saturated subsets are natural initial solutions of various algorithms applied in statistics and computational geometry; note that in some of the applications the set $\Fc$ is the result of a random process, in others it is a deterministic set obtained by an algebraic construction. Typically $n \gg m$.
\bigskip

To measure the quality of the subsets of $\Fc$, saturated or non-saturated, we consider an ``information-based'' $D$-optimality criterion, which is widely used in the theory of optimal experimental design (e.g., \cite{Fedorov}, \cite{Pazman86}, \cite{puk}, and \cite{AtkinsonEA07}).
\bigskip

 Define the information matrix of $\Ss \subseteq \Fc$ as $\M(\Ss)=\sum_{\f \in \Ss} \f\f'$. If $\Ss=\emptyset$, set $\M(\Ss)=\0_{m \times m}$. Then, the $D$-criterion is
\begin{equation}
 \phi:  2^\Fc \to [0,\infty); \: \phi(\Ss)={\det}^{1/m}(\M(\Ss)).
\end{equation}
The $D$-optimal $s$-element subset $\Ss^*$ is any subset of $\Fc$ that maximizes $\phi(\Ss)$ in the class of all $s$-element subsets of $\Fc$.
\bigskip

To make the problem meaningful, we assume that $\Fc$ is ``non-singular'' in the sense that $\M(\Fc)$ is a non-singular matrix, i.e., the $n \times m$ matrix $\F(\Fc)$ with rows $\f'_i$, $\f_i \in \Fc$, has full column rank. This assumption is clearly equivalent to the existence of a \emph{saturated} non-singular subset of $\Fc$. 
\bigskip

Note that for the case of saturated subsets, we can equivalently define $D$-optimality as follows: A saturated subset $\Ss^*$ is $D$-optimal if and only if it maximizes $|\det(\F(\Ss))|$. Geometrically, a $D$-optimal saturated subset corresponds to the maximum-volume parallelotope formed by an $m$-tuple of vectors from $\Fc$. For a general $s \geq m$, a $D$-optimal $s$-element subset of $\Fc$ maximizes the mean squared volume of the $\binom{s}{m}$ parallelotopes formed by the vectors of $\Ss$. This is a direct consequence of the Cauchy-Binet formula.
\bigskip

In the theory of optimal design of experiments, a $D$-optimal $s$-element subset is called a replication-free $D$-optimal exact experimental design of size $s$ (e.g., \cite{RaschEA}, cf. \cite{Fedorov89})\footnote{For the saturated case, the problem of $D$-optimal exact design with replication and the problem of $D$-optimal exact designs without replications are equivalent.}.  Computing such sets, or designs, is generally a difficult problem of discrete optimization; for a review of optimal design algorithms, see, e.g., Chapter 12 of \cite{AtkinsonEA07}, \cite{GoosJones} and \cite{MandalEA}.
\bigskip

Let $\Xi$ be the set of all probability measures on $\Fc$. Clearly, if we find some
\begin{equation}\label{eq:od}
\xi^* \in \mathrm{argmax}_{\xi \in \Xi} \: {\det}^{1/m} \left( \int_{\Fc} \f\f' d\xi(\f) \right),
\end{equation}
then we obtain the following upper bound on $\phi(\Ss^*)$:  
\begin{equation}\label{eq:ub}
\phi_s^* := s \: {\det}^{1/m} \left( \int_{\Fc} \f\f' d\xi^*(\f) \right).
\end{equation}
The probability measure $\xi^*$ from \eqref{eq:od} is called the $D$-optimal approximate design for the linear regression model with regressors $\Fc$.
The $D$-optimal approximate designs and, consequently, the bounds $\phi^*$, can be computed via a number of specialized algorithms of convex optimization (e.g., \cite{Yu11}, \cite{YangEA}, \cite{HarmanEA19}, cf. Chapter 9 of \cite{PronzatoPazman}).  
\bigskip

The $D$-efficiency of a set $\Ss_1$ relative to a non-singular set $\Ss_2$ is defined as $\phi(\Ss_1)/\phi(\Ss_2)$. If $\Ss^*$ is a $D$-optimal $s$-element subset, then the relative efficiency of an $s$-element subset $\Ss$ with respect to $\Ss^*$ will be denoted by $\mathrm{eff}(\Ss)$. Note that $\phi(\Ss)/\phi^*_s \leq \mathrm{eff}(\Ss) \leq 1$.
\bigskip

The algorithms for $D$-optimal subsets can also be used for various purposes other than constructing maximum-volume (configurations of) parallelepippeds or designs of statistical experiments. For instance, we may be interested in computing the ellipsoid $\mathrm{MVEE}_0(\mathcal{X})$\footnote{MVEE stands for ``Minimum-Volume Enclosing Ellipsoid''.} which has the minimum volume among all ellipsoids covering the set $\mathcal{X}=\{\x_1,\ldots,\x_n\}\subset \Rbb^d$ and centred in $\0_d$, or the ellipsoid $\mathrm{MVEE}(\mathcal{X})$ with the minimum volume among completely all ellipsoids covering $\mathcal{X}$. It turns out that the computation the MVEEs is equivalent to the computation of a $D$-optimal approximate design for the statistical model with regressors $\f_i=\x_i$ in the case of $\mathrm{MVEE}_0(\mathcal{X})$, and for the statistical model with regressors $\f_i=(\x_i',1)'$ in the case of $\mathrm{MVEE}(\mathcal{X})$. For details, see \cite{SilveyTitterington}, \cite{Titterington75}, \cite{Titterington78}, \cite{Ahipasaoglu15MVE}, \cite{Todd16} or the introduction of \cite{HarmanEA19}.  
\bigskip

Another application of $D$-efficient subsets is in the information-based data subselection. Consider a large sample consisting of the points $\x_1, \ldots, \x_n \in \Rbb^d$ (e.g., the vectors of covariates) and let $\f_i = f(\x_i)$, $i=1,\ldots,n$, be the corresponding regressors, where $f$ is an appropriate ``regression function''. Then, one can select an informative subsample of size $s$ by choosing $s$ regressors $\f_i$ (and the corresponding points $\x_i$) that make up a $D$-optimal or $D$-efficient subset (e.g., \cite{DrovandiEA} and \cite{WangEA}).
\bigskip

To compute an information-based subsample, \cite{WangEA} provides a fast specialized deterministic algorithm, primarily for sizes $s=2kd$, $k \in \mathbb{N}$. Note, however, that the algorithm, although fast and potentially very useful for big data, produces subsets that are not saturated, and generally not particularly $D$-efficient, in some cases even singular. The authors of \cite{DrovandiEA} also construct a subsample motivated by criteria used in experimental design, by solving two optimization problems: the first one on a possibly continuous superset of $\Fc$, and then second one on $\Fc$.  
\bigskip

As shown above, there is a large number of algorithms for constructing $D$-optimal or $D$-efficient subsets of a general size. However, these algorithms often require a small and non-singular initiation subset, preferably as efficient as possible. In particular, the saturated subsets can be used to initiate algorithms for both optimal \emph{approximate} designs and optimal \emph{exact} designs. In many such algorithms, such as \cite{Yu11}, \cite{YangEA}, and \cite{AtkinsonEA07}, the the initial solution suggested by the authors is either random or based on a regularized greedy heuristic (cf.  Subsections \ref{subsect:RND} and \ref{subsect:RGH}), which may not be the ideal method. It was the main motivation of this paper to analyse and compare alternative methods to generate such initial subsets.

\section{Greedy heuristics for saturated subsets}

A general form of the greedy procedure for the construction of saturated subsets is the following:
\smallskip

\noindent\textbf{Scheme of greedy heuristics for saturated subsets}\\
1. Set $j \leftarrow 0$ and $\Ss_0 \leftarrow \emptyset$\\
2. Select $\f_{(j)}$ from $\Fc \setminus \Ss_j$\:\: \textbf{(S)}\\
3. Set $\Ss_{j+1} \leftarrow \Ss_j \cup \{\f_{(j)}\}$\\
5. If $j+1<m$ set $j \leftarrow j+1$ and return to 2\\
6. Output $\Ss_m$
\bigskip

The main difficulty for the greedy construction of saturated subsets is that during the process of ``saturation'' we work with subsets with fewer than $m$ elements, and such subsets $\Ss$ are necessarily singular. That is, $\phi(\Ss)=0$ and the quality of such subsets cannot be directly compared via $\phi$. In particular, the general greedy heuristic, which in each step adds the regressor $\f_{(j)}$ that provides the highest increase in the objective function (e.g., \cite{NemhauserEA}) cannot be used. In \cite{Sagnol13}, the author circumvented this deficiency by considering optimality criteria that are not constantly zero for singular subsets, but we will not use this approach.
\bigskip

The methods described next solve the saturated subset problem in different ways, leading to various specifications of the critical step 2 which we call a saturation step and denote it by (S).  
\bigskip

The heuristic outlined above is sometimes called a ``forward'' method. However, there also exist ``drop'' or ``backward'' methods (e.g., \cite{RaschEA} and \cite{AtkinsonEA07}, Section 12.4), where we begin with the full set $\Fc$ and delete the least informative regressors one by one. These, however, turn out to be slow, because their asymptotic complexity is $O(n^2 - nm)$ and in applications $n$ is usually a large number.

\subsection{Random subsampling}\label{subsect:RND}

The simplest method for obtaining a saturated subset is to select $m$ regressors uniformly at random. Formally, in step (S) we select a random $\f_{(j)}$ from $\Fc \setminus \Ss_j$. This method is extremely fast, but the random saturated subsets tend to have low efficiency. It is even possible that such a saturated $\Ss$ is singular. The case of a singular random subset is not only possible in theory, but also in applications, as shown in Example \ref{exSingularRandom}.

\begin{example}\label{exSingularRandom}
	Consider the classical factorial linear regression model with $m$ factors, each with levels $-1,1$, with main factors and no constant term. In this situation the set of regressors is $\Fc = \{-1,1\}^m$. Suppose that we intend to construct a saturated experimental design at random. Then, the probability of obtaining a singular subset (i.e., a singular design) can be undesirably high. For example, for $m=3, 4, 5$, and $6$ this probability is $\frac{3}{7} \approx 0.429, \frac{223}{455} \approx 0.490, \frac{ 3285}{6293} \approx 0.522$, and $\frac{175795}{334707} \approx 0.525$.
\end{example}

To obtain more efficient subsets, methods for nonuniform random subsampling were proposed. In particular, the leverage-based random subsampling (e.g., \cite{MaEA15} and \cite{MaSun}) selects $\f_{(j)}=\f_i$ from $\Fc \setminus \Ss_j$ with probability $\pi_i$ that is proportional to its so-called leverage score $\f_i' \M^{-1}(\Fc) \f_i$.\footnote{Leverages based on subsets of $\Fc$, called sensitivities in optimal design of experiments (e.g., \cite{Fedorov}) naturally occur in most $D$-optimum design algorithms because they have an important analytic and statistical interpretation.} Note that a non-uniform subsampling can be much more computationally demanding than the uniform sampling. 
 Moreover, the leveraging methods are also not guaranteed to produce a non-singular subset, especially if one seeks a saturated subset: Example \ref{exSingularLev} demonstrates this.

\begin{example}\label{exSingularLev}
	In Example \ref{exSingularRandom}, the information matrix $\M(\Fc)$ is proportional to $\I_m$ and $\f_i'\f_i = m$ for all $i$. Therefore, all regressors have the same leverage score which means that the leveraging method is equivalent to the uniform random subset selection, including the probabilities of generating a singular subset as listed in Example \ref{exSingularRandom}.
\end{example}

A non-zero probability of generating a singular subset for a given $\Fc$ is not necessarily a damaging feature\footnote{Unless, of course, the probability of generating a singular subset is $1$ or very close to $1$.}; we can simply repeatedly use the method until we hit a non-singular subset. See Subsection \ref{subsect:multi} for further discussion on this ``multi-start'' approach.

\subsection{The regularized greedy heuristic (RGH)}\label{subsect:RGH}

Except for the simple or the weighted random choice of a saturated set $\Ss \subseteq \Fc$, the most common is the procedure based on a regularization of the information matrix of a sub-saturated set (see, e.g., Section 11.2 of \cite{AtkinsonEA07}; cf. \cite{Vuchkov}). The regularized greedy heuristic (RGH) specifies step (S) as
\begin{equation}\label{eq:RGH}
 \f_{(j)} \in \mathrm{argmax}_{\f \in \Fc\setminus\Ss_j}\phi(\M(\Ss_j) + \f\f' + \delta \I_m).
\end{equation}
Here, $\delta$ is a small positive constant; in the context of experimental design, \cite{AtkinsonEA07} suggest to take $\delta$ between $10^{-4}$ and $10^{-6}$. The matrix determinant lemma directly implies that \eqref{eq:RGH} is equivalent to 
\begin{equation}\label{eq:RGHvar}
  \f_{(j)} \in \mathrm{argmax}_{\f \in \Fc\setminus\Ss_j}\f'(\M(\Ss_j) + \delta \I_m)^{-1} \f.
\end{equation}

A short \texttt{R}-implementation of the RGH based on an efficient computation of \eqref{eq:RGHvar} is given in Appendix. The asymptotic complexity of this algorithm is $O(nm^3)$.
\bigskip

The numerical experience shows that the RGH usually finds a reasonably $D$-efficient saturated subset $\Ss$. However, a little known fact is that, similarly to the random methods, the RGH is not guaranteed to provide a non-singular saturated subset, even if one exists: 

\begin{example}
Consider the RGH with $\delta=10^{-4}$ for $\f_a = (1,0,0)'$, $\f_b=(0,1,0)'$, $\f_c=(1,1,0)'$ and $\f_d=(0,0,10^{-5})'$. Then, RGH first chooses $\f_c$, then $\f_a$ (or $\f_b$) and then $\f_b$ (or $\f_a$). The resulting subset $\Ss=\{\f_a,\f_b,\f_c\}$ is singular, although a non-singular subset $\{\f_a,\f_b,\f_d\}$ of size $3$ does exist.
\end{example} 

That is, we cannot fully rely on the non-singularity of the design produced by RGH. Moreover, unlike the random methods, the RGH is deterministic (cf. Subsection \ref{subsect:multi}), therefore it has no meaning to run it multiple times in search for an improvement. On the other hand, RGH can be also applied to criteria other than $D$-optimality in a straightforward manner. We will show in the following sections that to construct a $D$-efficient saturated subset, it may be preferable to use modifications of the original Galil-Kiefer and the Kumar-Yildirim methods.

\subsection{The Galil-Kiefer method (GKM)}

In \cite{GalilKiefer} (see also \cite{AtkinsonEA07}, Section 12.4) Galil and Kiefer proposed to specify step (S) as 
\begin{equation}\label{eq:GKM}
 \f_{(j)} \in \mathrm{argmax}_{\f \in \Fc\setminus\Ss_j}\prod_{i=1}^{j+1} \lambda_i(\M(\Ss_j) + \f\f'),
\end{equation}
where $\lambda_1(\M)\geq \cdots \geq \lambda_m(\M)$ denote the eigenvalues of a non-negative definite $m \times m$ matrix $\M$. That is, instead of maximizing the product of all eigenvalues $\phi(\Ss)$, which is zero, we maximize the product of the $j+1$ largest eigenvalues of $\M(\Ss)$, which can be non-zero. Let $\Ss \subseteq \Fc$ be an $r$-element set. Using the well known fact (e.g. 6.54 of \cite{Seber}) that the $r$ largest eigenvalues of $\M(\Ss)=\F'(\Ss)\F(\Ss)$ are equal to the $r$ eigenvalues of $\F(\Ss)\F'(\Ss)$, we see that \eqref{eq:GKM} is equivalent to 
\begin{equation}\label{eg:GKMdet}
 \f_{(j)} \in \mathrm{argmax}_{\f \in \Fc\setminus\Ss_j}\det\big(\F(\Ss_j \cup \f)\F'(\Ss_j \cup \f)\big),
\end{equation}
which was also noted in \cite{GalilKiefer}. Formula \eqref{eg:GKMdet} replaces the computation of eigenvalues by the simpler computation of determinants. 
\bigskip

Although the approach of Galil and Kiefer appears to be reasonable, it seems that it has not been much analysed in the literature and used in applications. This is probably caused by the fact that the determinant form \eqref{eg:GKMdet} is still computationally expensive. The paper \cite{GalilKiefer} does not provide a more efficient solution for the update of $\Ss_j$; the authors simply state that the standard formulas can be ``implemented with appropriate changes''.
\bigskip 

In the sequel, we show a more efficient and geometrically interpretable form of the optimization step. The basic insight is that the determinant form of the GKM has a simple geometric interpretation - we choose $\f_{(j)}$ to maximize the $j$-dimensional volume of the parallelotope spanned in $\Rbb^m$ by the  vectors of $\Ss_j \cup \f$, $\f \in \Fc \setminus \Ss_j$. Alternatively, we can formulate it as follows.

\begin{theorem}\label{tEquiv}
	The GKM always produces a non-singular saturated subset of $\Fc$ if such a subset exists. Moreover, the optimization \eqref{eq:GKM} (or \eqref{eg:GKMdet}) is equivalent to
	\begin{equation}\label{eg:GKMproj}
 \f_{(j)} \in \mathrm{argmax}_{\f \in \Fc\setminus\Ss_j} \f'(\I_m-\Pb(\Ss_j))\f,
\end{equation}\label{eq:GKMproj}
	where $\Pb(\Ss_j)$ is the projector on the column space of $\M(\Ss_j)$.
\end{theorem}
\begin{proof}
	If $\f \in \Cc(\F'(\Ss_j))$ then $\F(\Ss_j \cup \f)\F'(\Ss_j \cup \f)$ is not of full rank and its determinant is zero; here $\Cc$ denotes the column space. It follows that the GKM always chooses $\f \not\in \Cc(\F'(\Ss_j))$ and by induction, we obtain that the rank of $\M(\Ss_j)$ is $j$ for $j=1,\ldots,m$. Therefore, $\M(\Ss_j)$ can be expressed as $\V\V'$, where $\V$ is an $m \times j$ matrix of rank $j$. The positive eigenvalues of $\M(\Ss_j) + \f\f' = (\V,\f)(\V,\f)'$ are equal to the positive eigenvalues of $(\V,\f)' (\V,\f)$. That is,
	\begin{equation*}
	  \prod_{i=1}^{j+1} \lambda_i(\M(\Ss_j) + \f\f') = \det\begin{bmatrix}
	\V'\V & \V'\f \\ \f'\V & \f'\f
	\end{bmatrix},
   \end{equation*}
	which is equal to $(\f'\f - \f'\V(\V'\V)^{-1}\V'\f) \det(\V'\V)$ by the block determinant formula (e.g. \cite{Harville}, Theorem 13.3.8). Moreover, $\det(\V'\V)$ does not depend on $\f$ and $\V(\V'\V)^{-1}\V'=\Pb(\Ss_j)$.
\end{proof}

As a side result, we obtained that the GKM is guaranteed to produce a non-singular subset.

Clearly,
\begin{equation}
  \f'(\I_m-\Pb(\Ss_j))\f=\lVert (\I_m-\Pb(\Ss_j))\f \rVert^2,
\end{equation}
therefore Theorem \ref{tEquiv} implies that the algebraic approach based on the positive eigenvalues of the information matrix is geometrically a successive projection method: it projects the regressors to the orthogonal complement of $\Cc(\M(\Ss_j))$ and then finds the projected regressor of the largest norm. In an intuitive sense, such a regressor provides the most information not already captured by $\Ss_j$. 
\bigskip

Another observation is that the GKM is, roughly stated, a limit of the RGH as $\delta$ tends to zero. First, we need a technical lemma.

\begin{lemma}\label{lMaxFormula}
	Let $j \geq 2$, $\F=\F(\Ss_j)$, $\delta>0$. The set on the right hand side of \eqref{eq:RGH} is equal to 
	\begin{equation}\label{eq:RGHalt}
	 \mathrm{argmax}_{\f \in \Fc\setminus\Ss_j}\f'(\I_m-\F'(\F\F'+\delta\I_{j})^{-1}\F)\f.
	 \end{equation}
\end{lemma}

\begin{proof}
	Let $\M=\M(\Ss_j)$. Observe that $\lambda_i(\M + \f\f' + \delta \I_m) = \delta + \lambda_i(\F'\F + \f \f')$. The non-zero eigenvalues of $\F'\F + \f \f'$ coincide with the non-zero eigenvalues of $(\F',\f)'(\F',\f)$. It follows that $\det(\M + \f\f' + \delta \I_m)$ is equal to $\delta^{m-j-1}\prod_{i=1}^{j+1}(\delta+\lambda_i((\F',\f)'(\F',\f)))$, which can be expressed as
\begin{equation*}
	\delta^{m-j-1}\det(\begin{bmatrix}
	\F\F' & \F\f \\ \f'\F' & \f'\f
	\end{bmatrix}+\delta\I_{j+1}).
\end{equation*}
	The term $\delta^{m-j-1}$ can be ignored in the maximization, and using the formula for determinant of a block matrix, \eqref{eq:RGHalt} is equivalent to $\mathrm{argmax}_{\f \in \Fc\setminus\Ss_j} \det(\F\F' + \delta \I_{j})(\f'\f + \delta - \f'\F'(\F\F'+\delta \I_{j})^{-1}\F\f)$.
 Since the first term of the expression does not depend on $\f$, and the additive term $+\delta$ does not affect the maximization, we obtain the desired result.
\end{proof}

\begin{theorem}\label{tLimit}
	The GKM is the limit of the RGH in the sense that for any finite set $\mathcal{G} \subset \Rbb^m$, an $r \times m$ matrix $\F$, an $m \times m$ matrix $\M=\F'\F$, and the orthogonal projector $\Pb$ on $\mathcal{C}(\M)$ there is some $\delta_*>0$ such that 
	\begin{equation}\label{eq:rghsubgkm}
	\mathrm{argmax}_{\f \in \mathcal{G}}\phi(\M + \delta \I_m + \f \f') \subseteq \mathrm{argmax}_{\f \in \mathcal{G}} \f' (\I_m - \Pb) \f
	\end{equation}
	for all positive $\delta<\delta_*$.
\end{theorem}

\begin{proof}
  Let $h:=\max_{\f \in \mathcal{G}} \f' (\I_m - \Pb) \f$ and let $\epsilon>0$ be such that
  \begin{equation*}
  \mathrm{argmax}_{\f \in \mathcal{G}} \f' (\I_m - \Pb) \f = \{\f \in \mathcal{G}:  \f' (\I_m - \Pb) \f > h - \epsilon\}.
  \end{equation*}
  For any matrix $\A$, the pseudoinverse satisfies $\A^+ = \lim_{\delta\to 0^+} (\A'\A + \delta \I)^{-1}\A'$ (e.g. \cite{Harville}, Theorem 20.7.1) and $\F'(\F')^+=\Pb$, therefore
  \begin{equation*}
   \lim_{\delta\to 0^+}\f'(\I_m-\F'(\F\F'+\delta\I_r)^{-1}\F)\f = \f' (\I_m - \Pb) \f
   \end{equation*}
   for any $\f$. Since $\mathcal{G}$ is finite, there is some $\delta_*$ such that
   \begin{equation*}
   |\f'(\I_m-\F'(\F\F'+\delta\I_r)^{-1}\F)\f - \f' (\I_m - \Pb) \f| < \epsilon/2
   \end{equation*}
   for all $\delta < \delta_*$ and all $\f \in \mathcal{G}$. Let $\delta<\delta_*$ and let $\f_* \in \mathcal{G}$ be such that $\f_* \notin \mathrm{argmax}_{\f \in \mathcal{G}} \f' (\I_m - \Pb) \f$. Then $\f_*' (\I_m - \Pb) \f_* \leq h-\epsilon$, therefore $\f_*'(\I_m-\F'(\F\F'+\delta\I_r)^{-1}\F)\f_* < h - \epsilon/2$. It follows that $\f_* \notin \mathrm{argmax}_{\f \in \mathcal{G}} \f'(\I_m-\F'(\F\F'+\delta\I_r)^{-1}\F)\f_*$, i.e., $\f_* \notin \mathrm{argmax}_{\f \in \mathcal{G}} \phi(\M + \delta \I_m + \f \f')$ by the previous lemma.
\end{proof}

If the regressors $\f$ of $\mathcal{G}$ are the result of iid sampling from a continuous $m$-dimensional  distribution, then the set on the RHS of \eqref{eq:rghsubgkm} is a singleton with probability $1$, i.e., both sides of \eqref{eq:rghsubgkm} coincide. However, for a symmetric $\mathcal{G}$, the right hand side can be a large set, and the solutions of the RGH can form a strict subset of those of GKM for any $\delta>0$; see Section \ref{sec:num}, specifically panels (a), (b), (d) in Figure \ref{fTimeEffcube}.
\bigskip 

The GKM is not only guaranteed to produce a non-singular saturated subset, it is guaranteed to produce a subset with efficiency at least $1/m$, as we show next. First, a technical lemma.

\begin{lemma}\label{lEff}
	Let $\f_1, \ldots, \f_m \in \Rbb^m$. Then, $\det(\sum_j \f_j \f_j') = \det(\sum_j \tilde{\f}_j\tilde{\f}_j')$, where $\tilde{\f}_1 = \f_1$ and $\tilde{\f}_j$ for $j>1$ can be defined by any of the three equivalent formulas:
	\begin{align}
	\tilde{\f}_j
	&= (\I - \Pb_{1:(j-1)})\f_j \label{eq:tildeF1} \\
	&= (\I - \tilde{\Pb}_{j-1} -  \ldots - \tilde{\Pb}_{1})\f_j, \label{eq:tildeF2} \\
	&= (\I-\tilde{\Pb}_{j-1}) \cdots (\I - \tilde{\Pb}_{1})\f_j \label{eq:tildeF3}.
	\end{align}
	In the expressions above $\Pb_{1:(j-1)}$ is the projection matrix to $\Cc(\f_1, \ldots, \f_{j-1})$ and $\tilde{\Pb}_{i}$ is the projection matrix to $\Cc(\tilde{\f}_i)$.
\end{lemma}

\begin{proof}
	From formula \eqref{eq:tildeF2}, it is clear that the transformation from $\F=(\f_1, \ldots, \f_m)$ to $\tilde{\F}=(\tilde{\f}_1, \ldots, \tilde{\f}_m)$ is in fact a Gram-Schmidt orthogonalization, where the obtained orthogonal vectors $\tilde{\f}_j$ are not normalized. That is, $\F=\tilde{\F}\Rb$ for some upper triangular matrix $\Rb$ with ones on the diagonal. It follows that $\det(\F) = \det(\tilde{\F})$, which yields the desired result $\det(\F\F') = \det(\tilde{\F}\tilde{\F}')$.
	
	Since $\tilde{\f}_1, \ldots, \tilde{\f}_m$ are orthogonal, we have $\tilde{\Pb}_i \tilde{\Pb}_k = \0$ for any $i \neq k$, and hence \eqref{eq:tildeF3} is obtained. Formula \eqref{eq:tildeF1} holds because $\Cc(\f_1, \ldots, \f_{j-1}) = \Cc(\tilde{\f}_1, \ldots, \tilde{\f}_{j-1})$.
\end{proof}

\begin{theorem}
  Let $\Ss_{GK}$ be the saturated subset of $\Fc$ produced by the GKM. Then $\mathrm{eff}(\Ss_{GK}) \geq \frac{1}{m}$.
\end{theorem}
\begin{proof}
Let $\Ss_{GK}=\{\f_1, \ldots, \f_{m}\}$ and define $\tilde{\Ss} = \{\tilde{\f}_1, \ldots, \tilde{\f}_m\}$, where the $\tilde{\f}_j$'s are given by Lemma \ref{lEff}. Then $\phi(\Ss_{GK}) = \phi(\tilde{\Ss})$. Let us also work with ``normalized'' information matrix $\N(\Ss) = \sum_{\f \in \Ss} \f\f'/s$, where $s$ is the size of $\Ss$. Without loss of generality, let us rotate the coordinate axes of $\Rbb^m$, so that $\tilde{\f}_j = \Vert \tilde{\f}_j \Vert \e_j$, which is possible because $\tilde{\f}_1, \ldots, \tilde{\f}_m$ are orthogonal. Then, $\det(\sum_i \tilde{\f}_i\tilde{\f}_i') = \prod_i \Vert \tilde{\f}_i \Vert^2$. Moreover, the set $\Ss_B$ that consists of all points of the form $(\pm \Vert \tilde{\f}_1 \Vert, \ldots, \pm \Vert \tilde{\f}_m \Vert)'$ satisfies $\det(\N(\Ss_B)) = \prod_i \Vert \tilde{\f}_i \Vert^2$ and $\N(\Ss_B) \succeq \N(\Ss)$ for any $\Ss \subseteq \Fc$ because all the elements of $\Fc$ lie in $\mathrm{conv}(\Ss_B)$. In particular, $\N(\Ss_B) \succeq \N(\Ss^*)$, where $\Ss^*$ is the $D$-optimal subset of size $m$. It follows that $\det(\M(\Ss^*)) = m^m\det(\N(\Ss^*)) \leq m^m\det(\N(\Ss_B)) = m^m \prod_i \Vert \tilde{\f}_i \Vert^2$. Therefore,
$$
\frac{ \phi(\Ss_{GK}) }{ \phi(\Ss^*) } \geq \frac{ \phi(\tilde{\Ss}) }{ \phi(\Ss_{B}) } = \frac{ \prod_i \Vert \tilde{\f}_i \Vert^{2/m} }{m \prod_i \Vert \tilde{\f}_i \Vert^{2/m} } = \frac{1}{m}.
$$
\end{proof}

As a side result, we obtain an explicit formula for the $D$-optimality value of $\Ss_{GK}$: $\phi(\Ss_{GK}) = \prod_i \Vert \tilde{\f}_i \Vert^{2/m} $, where the projected regressors $\tilde{\f}_i$ are constructed during the GKM process.
\bigskip

A na\"{i}ve implementation of the GKM using \eqref{eg:GKMproj} is relatively time consuming as it requires the calculations of projection matrices $\Pb(\Ss_j)$ and the projections of the $n$ regressors. However, \eqref{eg:GKMproj} depends only on the norms of the projected regressors $\tilde{\f}$ and formula \eqref{eq:tildeF3} shows that the projected regressors can be calculated successively by applying projections to the one-dimensional subspaces orthogonal to $\tilde{\f}_j$. 
Therefore, the algorithm can ``forget'' the original regressors, which allows for the following efficient implementation of the GKM:
\bigskip

\noindent\textbf{Efficient GKM}\\
1. Set $\ind \leftarrow \emptyset$\\
2. $j \leftarrow \arg\max_{i=1,\ldots,n} \lVert \f_i\rVert$\\
3. Set $\ind \leftarrow \ind \cup \{j\}$\\
5. If $j<m$, update $\f_i \leftarrow \f_i -  [\f_j' \f_i / (\f_j'\f_j)] \f_j $, set $j \leftarrow j+1$ and return to 2\\
6. Output $\ind$
\bigskip

The efficient formulation of the GKM has asymptotic complexity $O(nm^2)$. Its implementation in the programming language of \texttt{R} is in Appendix.
\bigskip

The GKM is non-random which is, in some situations, a drawback.\footnote{For instance, if we use the saturated subsets for initiation of deterministic heuristics, a randomized initiation method is crucial for the approach of multiple restarts.} Galil and Kiefer were aware of this problem and they proposed a randomisation based on first selecting some points at random and only then starting to apply the GKM. However, such an approach may lead to singular $m$-point subsets, if the first random points were chosen ``unluckily''.
\bigskip

Nevertheless, the GKM can be also randomised as follows: in step (S), instead of selecting regressor $\f_{(j)}$ that maximizes $v(\f) := \f' (\I_m - \Pb(\M(\Ss_j)) \f$, we may select $\f_{(j)}$ at random, with probabilities of the particular regressors $\f$ proportional to $v^\alpha(\f)$, where $\alpha>0$; the higher $\alpha$ is, the less random is the procedure. 

Because the randomised GKM never chooses regressors that lie in $\mathrm{Span}(\f_1, \ldots, \f_{j-1})$, it also arrives at a non-singular subset if such a subset exists by the same argument as in the non-random case.

\subsection{The Kumar-Yildirim method (KYM)}

Originally, the Kumar-Yildirim method was proposed in \cite{KumarYildirim} to find initial points for computing the $\mathrm{MVEE}$, which need not be centred at the origin. In each iteration, the method chooses $2$ points $\f_{j_1}, \f_{j_2} \in \Fc$ maximizing and minimizing the scalar product $\f'\b$ for a randomly chosen $\b$. Thus, the output is a set of $2m$ points. 
\bigskip

But the $D$-optimal approximate design problem is equivalent to computing the $\mathrm{MVEE}_0$ centred at the origin and containing $\f_1,\ldots,\f_n$. For such a problem, one can apply the original Kumar-Yildirim method with formally added points $-\f_1, \ldots, -\f_n$. Then, the method selects in each step both $\f_j$ and $-\f_j$ for some $j$, which is redundant. It follows that it is enough to choose one of the two points.
\bigskip

Hence, the modified KYM can be stated as a special case of the general greedy heuristic where step (S) is 
\begin{equation*}
  \f_{(j)} \in \mathrm{argmax}_{\f \in \Fc\setminus\Ss_j} |\f'\b|
\end{equation*}
and $\b$ is a random direction from $\Cc^\perp(\Ss_j)$.

It is not difficult to see that the output of the KYM is a non-singular saturated subset. However, we can formulate a stronger statement as follows.

\begin{theorem}\label{tKYeff}
  Let $\Ss_{KY}$ be the saturated subset of $\Fc$ produced by the KYM. Then
  \begin{equation*}
    \mathrm{eff}(\Ss_{KY}) \geq \frac{\pi}{4m(\Gamma(1+m/2))^{2/m}} \geq \frac{\pi}{2m^2}.
  \end{equation*}
\end{theorem}
\begin{proof}
The authors of \cite{KumarYildirim} obtain the efficiency bound for the output of their algorithm by applying the results of the paper \cite{BetkeHenk}. There, it is shown that $\vert\det(\z_1, \ldots, \z_m) \vert / m! \leq \vol(K) \leq \vert \det(\z_1, \ldots, \z_m) \vert$, where each $\z_j$  is equal to the difference $\f_{j_1} - \f_{j_2}$ as in the KYM, but for a general convex body $K$. By considering the convex body $K=\mathrm{conv}(\pm\f_1, \ldots, \pm\f_n)$\footnote{$K$ is called the Elfving set in the optimal experimental design literature}, we obtain for  $\Ss_{KY}$ that
\begin{equation*}
 \vol(K) \leq \vert \det(2\F(\Ss_{KY})) \vert = 2^m \det(\M(\Ss_{KY}))^{1/2},
\end{equation*}
where $2\F(\Ss_{KY})$ appears in the inequality because in each iteration of the algorithm in \cite{BetkeHenk}, one chooses $\z_j = \f_{(j)} - (-\f_{(j)}) = 2\f_{(j)}$ for the $\f_{(j)}$ maximizing $\vert \f_i' \b \vert$.

Let $\MVEE^* =\MVEE(K)$. In \cite{John}, John showed that $K \supseteq \MVEE(K)/\sqrt{m}$ if $K$ is centred at the origin; hence $\vol(K) \geq \vol(\MVEE^*)/m^{m/2}$ in our settings. Combining the previous results, we obtain that
\begin{equation*}
\begin{aligned}
\frac{1}{m^{m/2}} \vol(\MVEE^*) 
&\leq 2^m (\det(\M(\Ss_{KY}))^{1/2}.
\end{aligned}
\end{equation*}

From the results mentioned, e.g., in \cite{KumarYildirim} it follows that
$$ 
\vol(\MVEE^*) = \eta (\det \M(\Ss^*))^{1/2},
$$ 
where $\Ss^*$ is the $D$-optimal saturated subset and $\eta =  \pi^{m/2} / \Gamma(1+m/2)$ is the volume of the unit ball. Then,
$$
\frac{\eta (\det \M(\Ss^*))^{1/2}}{m^{m/2}} \leq 2^m (\det(\M(\Ss_{KY}))^{1/2}
$$
and
$$
\frac{\Phi_D(\M(\Ss_{KY}))}{\Phi_D(\M(\Ss^*))} \geq \frac{\eta^{2/m}}{4m} = \frac{\pi}{4m(\Gamma(1+m/2))^{2/m}}.$$
\end{proof}

For a large $m$, the Stirling approximation of $\Gamma$ yields
\begin{equation*}
 \mathrm{eff}(\Ss_{KY}) \geq \frac{\pi}{4m(\Gamma(1+m/2))^{2/m}} \approx \frac{e\pi}{2m^2}.
\end{equation*}
Note that by a direct application of the efficiency limit provided by Lemma 3.1 of \cite{KumarYildirim}, one would obtain a weaker bound $\mathrm{eff}(\Ss_{KY}) \geq m^{-3}$ than that in Theorem \ref{tKYeff}. The bound in Theorem \ref{tKYeff} is tighter because Kumar and Yildirim deal with general sets, and the centrally symmetric sets considered in this paper allow for stronger results.
\bigskip

An efficient implementation of the KYM in the programming language of \texttt{R} is given in Appendix. This implementation has asymptotic complexity $O(nm^2)$, i.e., the same as the GKM implementation, although non-asymptotically the KYM tends to be faster than GKM (unless $m \gtrapprox n/2$). 
\bigskip

Interestingly, the main idea of the KYM is similar to the projective interpretation of the GKM: both procedures seek to find a ``good'' new regressor in the space $\Cc(\M(\Ss))^\perp$ not yet spanned by the chosen regressors. However, the respective approaches for measuring what is a good regressor differ.
\bigskip

Note that both GKM and KYM are invariant under rotation, and, unlike the RGH, they are also invariant with respect to a centred dilation $\f \to c\f$, $c>0$, i.e., invariant under the common change of the units of measurement, if the regressors correspond to explanatory variables.

\section{Variants and extensions}
   
\subsection{Subsets of sizes other than $m$}

  The greedy heuristics for saturated subsets can be easily modified to generate efficient non-saturated subsets, i.e., subsets of sizes $s$ other than $m$.
  \bigskip 

The problem of obtaining efficient supersaturated experimental designs (e.g., \cite{Gilmour06}) is very similar to problem of constructing $D$-efficient subsets with $s<m$. In fact, the main idea of the Galil-Kiefer initiation for $D$-optimality, i.e. that $\det(\F(\Ss)\F'(\Ss))$ can be maximized instead of $\det(\F'(\Ss)\F(\Ss))$, is the same as the approach to $D$-optimal supersaturated designs in \cite{JonesMajumdar}. Note, however, that the motivation of Jones and Majumdar in  \cite{JonesMajumdar} for working with $\F(\Ss)\F'(\Ss)$ arises from a minimum bias estimator of the model parameters, and that they provide analytical results on optimal supersaturated designs in main-effects model, not design algorithms.
 \bigskip

Based on a greedy heuristic for saturated subsets, we can formulate a general scheme for a greedy heuristic for arbitrary size subsets $s$. Simply stated, if $s<m$ we can pre-maturely stop the greedy heuristics and if $s>m$, we can repeatedly perform the heuristic for a subset of size $m$ (or less than $m$ in the last run) and after each run of the method remove the $m$ obtained regressors from $\Fc$. We can generalize this idea such that at each selection step, we select not a single regressor but a batch of $b$ regressors, which can be computationally advantageous. More precisely:
\bigskip

\noindent\textbf{Scheme of greedy heuristics for subsets of size $s$}\\
1. Set $j \leftarrow 0$ and $\Ss_0 \leftarrow \emptyset$\\
2. Set $a \leftarrow \min(s-j,b)$\\
3. Select an $a$-element subset $\Fc_j$ of $\Fc \setminus \Ss_j$\:\: \textbf{(S)}\\
4. Set $\Ss_{j+a} \leftarrow \Ss_j \cup \Fc_j$\\
5. If $j<s$ set $j \leftarrow j+a$ and return to 2.\\
6. Output $\Ss_s$.
\bigskip

Of course here we will also need to keep and properly update some parameters which regulate step (S). Special case $b=1$, $s=m$ provides the scheme of the heuristics for saturated subsets, but note that the IBOSS mentioned in the introduction and the original KYM are also special cases of this general scheme. The multitude of various specifications of the above scheme and the comparison of the resulting algorithms is left for further research.

\subsection{A pre-selection strategy for a big $n$}\label{subsect:presel}

Let $\mathcal{R}:=\{\g_1,\ldots,\g_{km}\}$ be a uniformly randomly selected $km$-element subset of $\Fc$, where $1 \leq k \leq n/m$. Let us estimate the probability that $\mathcal{R}$ is non-singular. Let $\h_1,\ldots,\h_{km}$ be uniformly randomly chosen from $\Fc$ with replacement. For $t=0,\ldots,k-1$ let $A_t$ be the event that $[\h_{tm+1},\ldots,\h_{(t+1)m}]$ is a singular $m \times m$ matrix. Clearly, $A_0,\ldots,A_{k-1}$ are independent, therefore
\begin{eqnarray*}
 P[\mathcal{R} \text{ is non-singular}] =  P[\g_1,\ldots,\g_{km} \text{ span } \Rbb^m] \geq \\
 P[\h_1,\ldots,\h_{km} \text{ span } \Rbb^m] \geq
 1-P[\cap_t A_t] \geq 1-p^k,
\end{eqnarray*} 

where $p$ is the probability of each of $A_t$. That is, the probability that a $km$-element uniformly selected random subset of $\Fc$ is singular vanishes at least exponentially with respect to $k$.
\bigskip

At the same time, the probability $p$ is itself not too close to $1$ for most of sets $\Fc$ used in applications. Even for such a symmetric set as $\Fc = \{-1,1\}^m$ from Example \ref{exSingularRandom} it is possible to show that $p \leq (3/4 + O(1))^m$; cf. \cite{TaoVu}.
\bigskip

This motivates the following strategy to obtain an efficient saturated subset of $\Fc$:
\begin{enumerate}
  \item Randomly pre-select $\tilde{\Fc}$ of size $km$ from $\Fc$.
  \item Use either GKM or KYM to construct a saturated subset of $\tilde{\Fc}$.
\end{enumerate}
Note that the size of $\tilde{\Fc}$ does not depend on $n$, i.e., it can be small also for an enormous set $\Fc$. Even if $k$ is only moderately large, the probability that the resulting saturated subset is singular can be negligibly small. For instance, if $k=50$ and $p=3/4$, the probability of a singular resulting subset is less than $10^{-6}$.
\bigskip

We will show in Section \ref{sec:num} that the pre-selection strategy is a good compromise between the speed and efficiency.

\subsection{Multi-run approach}\label{subsect:multi}

 Let H be a randomized heuristic. Since H can generate subsets of various efficiencies, a natural idea is to run H multiple times and then select the best subset obtained. Note that in this respect the randomness inherent in the result can be a positive, not a negative feature of an algorithm.
\bigskip 
 
 Denote $Y_i$, $i=1,2,\ldots$, the iid random variables representing the efficiencies of the subsets obtained by independent runs of H. Let $0<\mathrm{eff} \lesssim 1$ be such that $P[Y_1\geq\mathrm{eff}]>0$. Then it can be shown that we need roughly $\ln(2)/P[Y_1 \geq \mathrm{eff}]$ runs of H for the median of $\max_{1 \leq i \leq N}Y_i$ to be $\mathrm{eff}$. In this sense, two heuristics are similarly good, if they have similar ratios $t^{-1}P[Y_1 \geq \mathrm{eff}]$, where $t$ is the time of a single run of H plus the time needed to compute the quality of the resulting subset.
\bigskip
 
The performance of the multi-run approach is therefore given by the speed of H as well as the tails of the distribution of $Y_1$, and both quantities may be difficult to assess. However, numerical simulations for the randomized heuristics that we consider in this paper suggest that the variability of $Y_1$ tends to be rather small, therefore $P[Y_1 \geq \mathrm{eff}]$ is mostly determined by the median of $Y_1$. Hence, in the assessment of the performance of a heuristic, the median efficiency of the subsets resulting from a single run is the more important factor than the speed. See Section \ref{sec:num} for numerical results.

\subsection{Eigenvalue-based criteria}

We say that a criterion $\phi: 2^{\Fc} \to \Rbb$ is eigenvalue-based if it depends only on the eigenvalues of $\M(\Ss)$. For instance, $D$-optimality can be expressed as $\phi(\Ss)=\left(\prod_i \lambda_i(\M(\Ss))\right)^{1/m}$. Besides $D$-optimality, the second most important eigenvalue-based criterion is $A$-optimality defined as $\phi_A(\Ss)=\left(\sum_i \lambda_i^{-1}(\M(\Ss))\right)^{-1}$ for a non-singular $\M(\Ss)$ and $\phi_A(\Ss)=0$ for a singular $\M(\Ss)$ (e.g., \cite{puk}, Chapter 6).
\bigskip

Clearly, the basic idea of the Galil-Kiefer method can be used for any eigenvalue-based criterion -- we can simply specify (S) as selecting $\f_{(j)}$ that maximizes an appropriate function of the non-zero eigenvalues of $\M(\Ss_j) + \f\f'$. In some cases, like $D$-optimality, such maximization can be performed much more easily than via numerical evaluation of the eigenvalues. For an example, we can utilize the following theorem for $A$-optimality: 

\begin{theorem}
	Let $\M(\Ss)$ be of rank $r$. Then, 
	\begin{eqnarray}\label{eAoptMax}
		\mathrm{argmax}_{\f \in \Fc\setminus\Ss}\left(\sum_{i=1}^{r+1}\lambda_i^{-1}(\M(\Ss) + \f\f')\right)^{-1}= \nonumber \\
		\mathrm{argmax}_{\f \in \Fc\setminus\Ss}\frac{\f'(\I_m-\Pb(\Ss))\f}{1+\f'\M^+(\Ss)\f},
	\end{eqnarray}
\end{theorem}

\begin{proof}
	The positive eigenvalues of $\M(\Ss) + \f\f' = (\V,\f)(\V,\f)'$ are equal to the positive eigenvalues of $(\V,\f)' (\V,\f)$ with the notation as in the proof of Theorem \ref{tEquiv}. Hence
	$$\sum_{i=1}^{r+1} \lambda_i^{-1}(\M(\Ss) + \f\f') = \tr\begin{bmatrix}
	\V'\V & \V'\f \\ \f'\V & \f'\f
	\end{bmatrix}^{-1},$$
	which is equal to $\tr(\V'\V)^{-1} + s^{-1}\f'\V(\V'\V)^{-2}\V'\f + s^{-1}$, where $s=\f'(\I_m - \V(\V'\V)^{-1}\V')\f = \f'(\I_m - \Pb(\M(\Ss)))\f$, based on the formula for the inverse of a block matrix (\cite{Harville}, Theorem 8.5.11). Observing that $\V(\V'\V)^{-2}\V'=\M^+$ yields that the studied maximization problem is equivalent to the minimization of
	\begin{equation}\label{eAoptMin}
		\tr(\V'\V)^{-1} + \frac{1+ \f'\M^+\f}{\f'(\I_m - \Pb_\M)\f}.
	\end{equation}
	Because $\tr(\V'\V)^{-1}$ does not depend on $\f$, the minimization of \eqref{eAoptMin} is equivalent to the maximization of \eqref{eAoptMax}.
\end{proof}

That is, the GKM for $A$-optimality is very similar to the GKM for $D$-optimality, but $A$-optimality requires an additional ``correction term'' to be added to \eqref{eq:GKMproj}.
\bigskip

It is unclear how can the KYM be tailored to $A$-optimality and other eigenvalue-based criteria. However, the efficiency of the $D$-optimal approximate design with respect to a very general class of eigenvalue-based criteria is bounded from below by $1/m$ (\cite{Harman04Moda}) which suggests that the KYM (as well as the GKM for $D$-optimality) will produce reasonably efficient subsets even in the case of other criteria than $D$-optimality. This is also supported by our numerical experience. We leave the development of the heuristics for general eigenvalue-based criteria to further research.

\section{Numerical comparisons}\label{sec:num}

To provide a brief but comprehensive comparison of the performances of the greedy methods we generate the sets $\Fc$ at random, as well as deterministically: 
\begin{enumerate}
 \item We first generate a covariance matrix $\Sigma$ from the Wishart distribution $\mathcal{W}_d(\I_d,d)$ and then $\x_i$ independently from $\mathcal{N}_d(\0_d,\Sigma)$. We set $\Fc=\{(\x'_1,1)', \ldots, (\x'_n,1)'\} \subset \Rbb^m$, $m=d+1$.
 \item We set $\Fc=\{-1,1\}^m$.
\end{enumerate}

The efficiencies of the saturated subsets $\Ss$ are calculated as $\phi(\Ss)/\phi^*_m$, where $\phi^*_m$ is computed via the REX algorithm suggested in \cite{HarmanEA19}. We used a computer with a 64-bit Windows 10 system running an Intel Core i7-5500U CPU processor at $2.40$ GHz with $8$ GB of RAM. For the RGH, we used the regularization parameter $\delta=10^{-4}$.
\bigskip

The results are exhibited in Figures \ref{fTimeEffnorm} and \ref{fTimeEffcube}. The numerical observations can be summarized as follows:
\begin{itemize}
\item Both RGH and GKM tend to provide highly efficient subsets. For a random $\Fc$ the results of GKM and RGH usually coincide and for $\Fc=\{-1,1\}^m$, RGH sometimes gives slightly more efficient subsets than GKM. On the other hand, GKM tends to be as fast as or faster than RGH. 
\item The KYM is somewhat faster but less efficient than GKM.
\item As one might expect, the random initiation (RND) is the fastest and the least efficient method. The leverage-weighted random sampling (RNDl), performs badly in all analysed cases. For $\Fc=\{-1,1\}^m$ both RND and RNDl occasionally fail to the extent that they provide singular subsets. 
\item GKM and KYM applied to a random subsample of size $50m$ (cf. Subsection \ref{subsect:presel}), denoted by GKMf and KYMf in the figures, are both very fast and efficient. In some cases, the speed of GKMf and KYMf is only slightly smaller than the speed of the random sub-sampling methods, yet, with a few exceptions, the efficiency is the same or almost the same as that of GKM and KYM applied to the full set $\Fc$. Moreover, the computation times and the resulting efficiencies are relatively stable under changes of $n$ and $m$.
\end{itemize}

To compare the degree of the efficiency improvement using the multi-run approach described in Subsection \ref{subsect:multi}, we computed the ``time profiles'' of the increase in efficiency of the best subset found by a repeated independent application of the individual heuristics. Evidently, for all procedures the efficiency increases very slowly (notice that the horizontal axis denotes the \emph{logarithm} of the time), i.e., for the performance of the multi-run method the efficiency of a single-run is crucial. Naturally, the numerical differences between the studied methods depend on a multitude of factors, such as the actual implementation used and the set $\Fc$ itself. Nevertheless, the exhibited figures are representative of the general behaviour of the heuristics that we observed. 

\begin{figure}[!h]
\centering
\includegraphics[width=\linewidth]{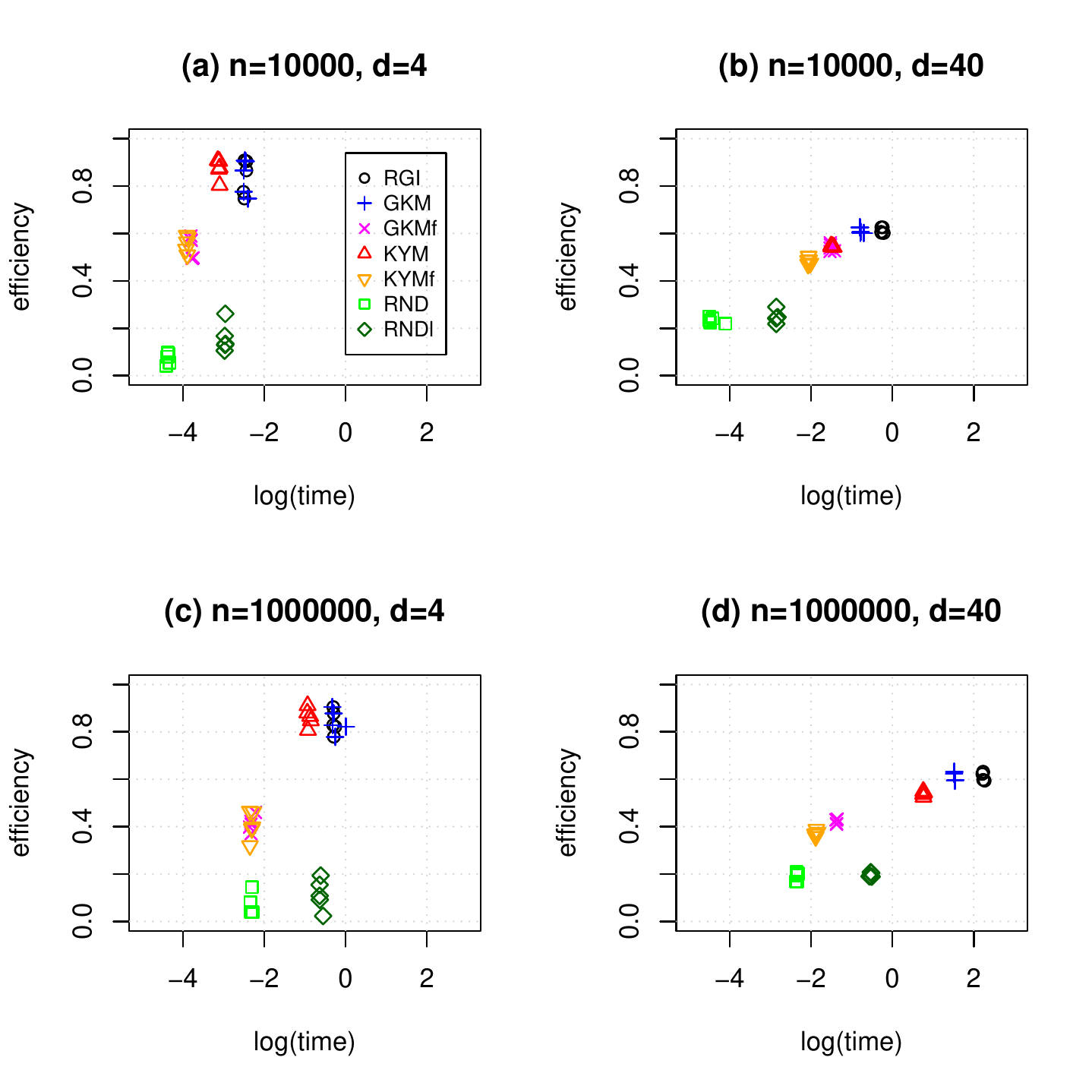}
\caption{The running times (in decadic logarithm) of the greedy saturated subsample heuristics and the $D$-efficiencies of the resulting subsets. Regressors generated from a lifted $d$-dimesional normal distribution ($m=d+1$) as detailed in the text.}\label{fTimeEffnorm}
\end{figure}

\begin{figure}[!h]
\centering
\includegraphics[width=\linewidth]{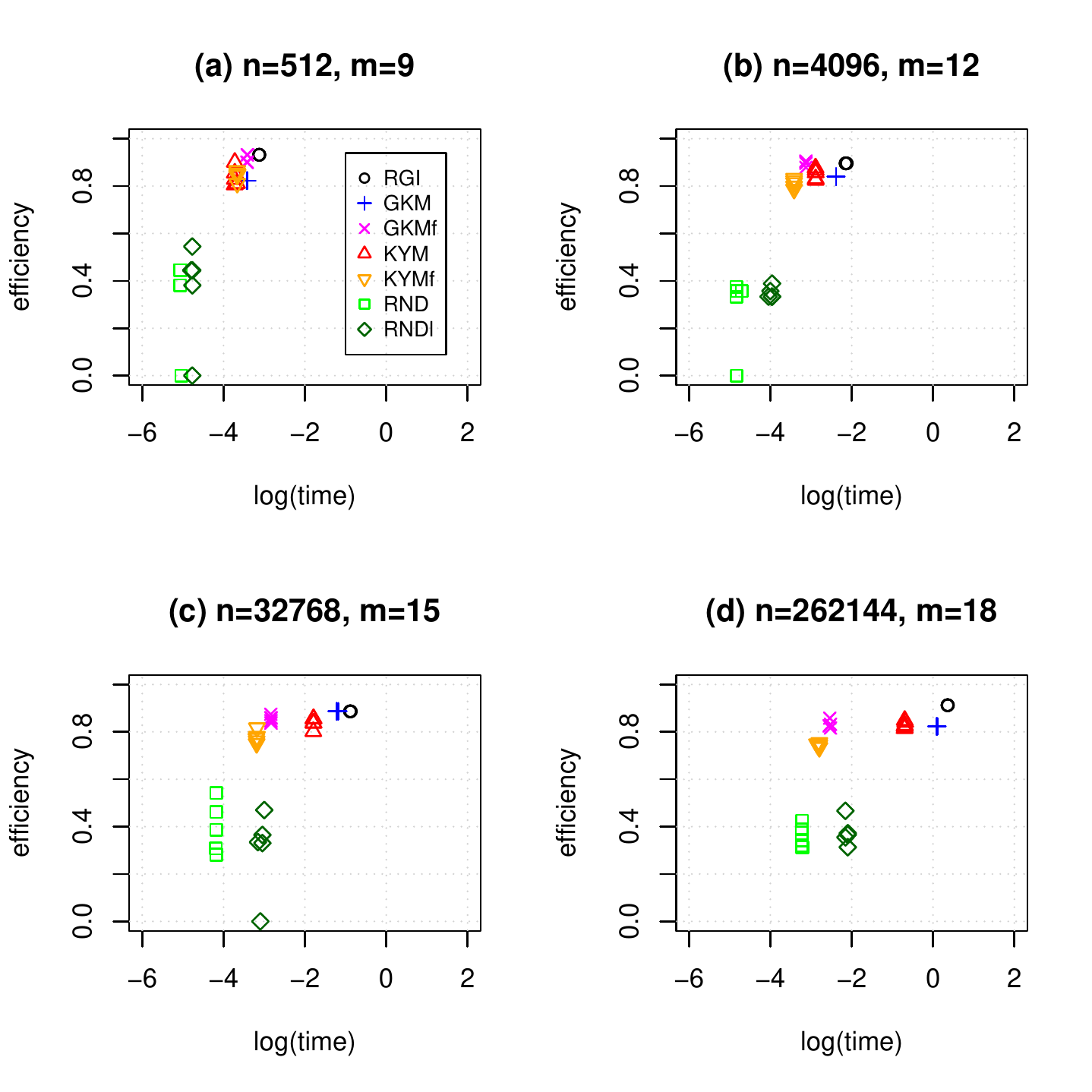}
\caption{The running times (in decadic logarithm) of the greedy saturated subsample heuristics and the $D$-efficiencies of the resulting subsets. Regressors are deterministically set to $\Fc=\{-1,1\}^m$.}\label{fTimeEffcube}
\end{figure}

\begin{figure}[!h]
\centering
\includegraphics[width=0.8\linewidth]{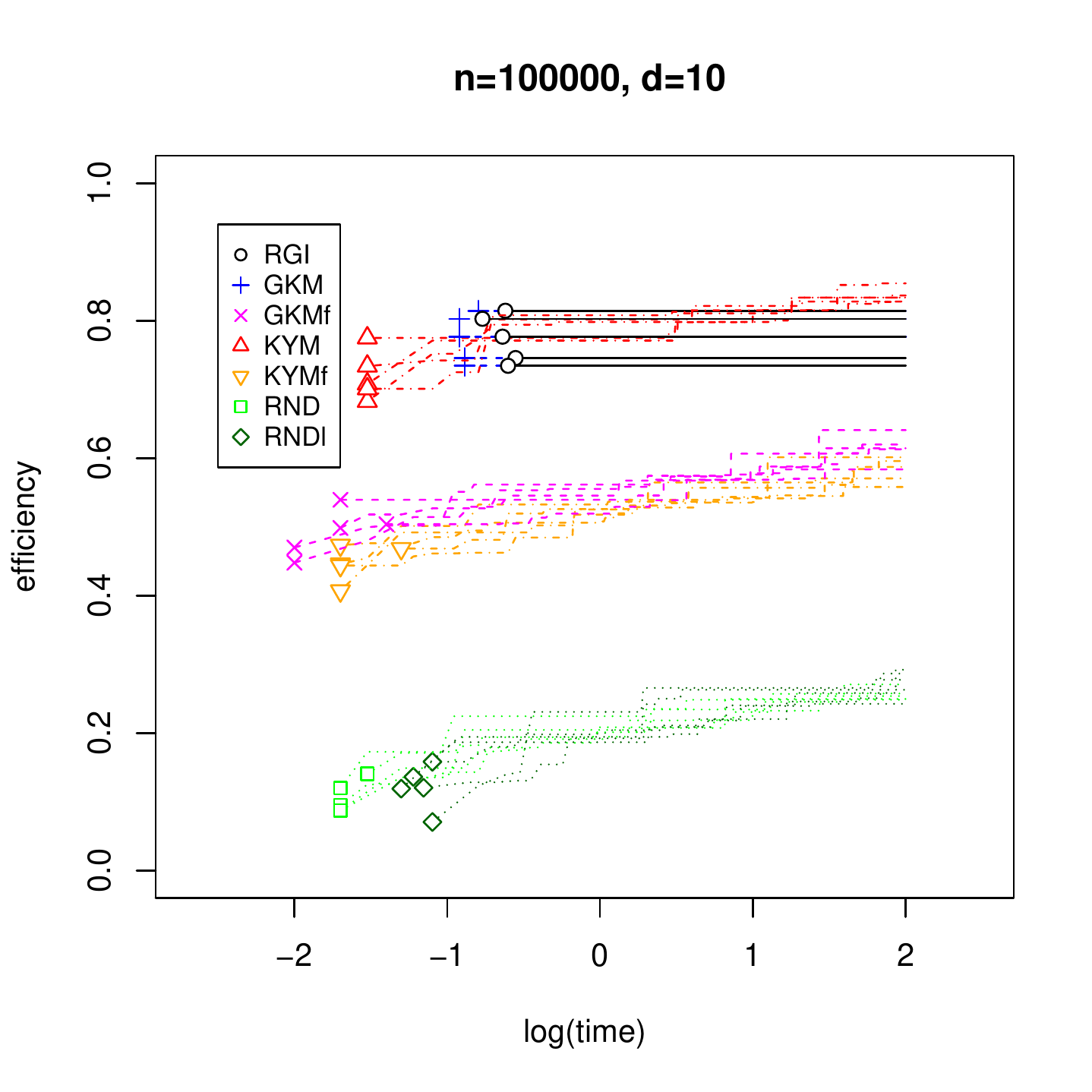}
\caption{The dependence of the best efficiency on the computation time achieved by the multistart approach. Regressors generated from a lifted $d$-dimesional normal distribution ($m=d+1$) as detailed in the text.}\label{ProfileNorm}
\end{figure}

\begin{figure}[!h]
\centering
\includegraphics[width=0.8\linewidth]{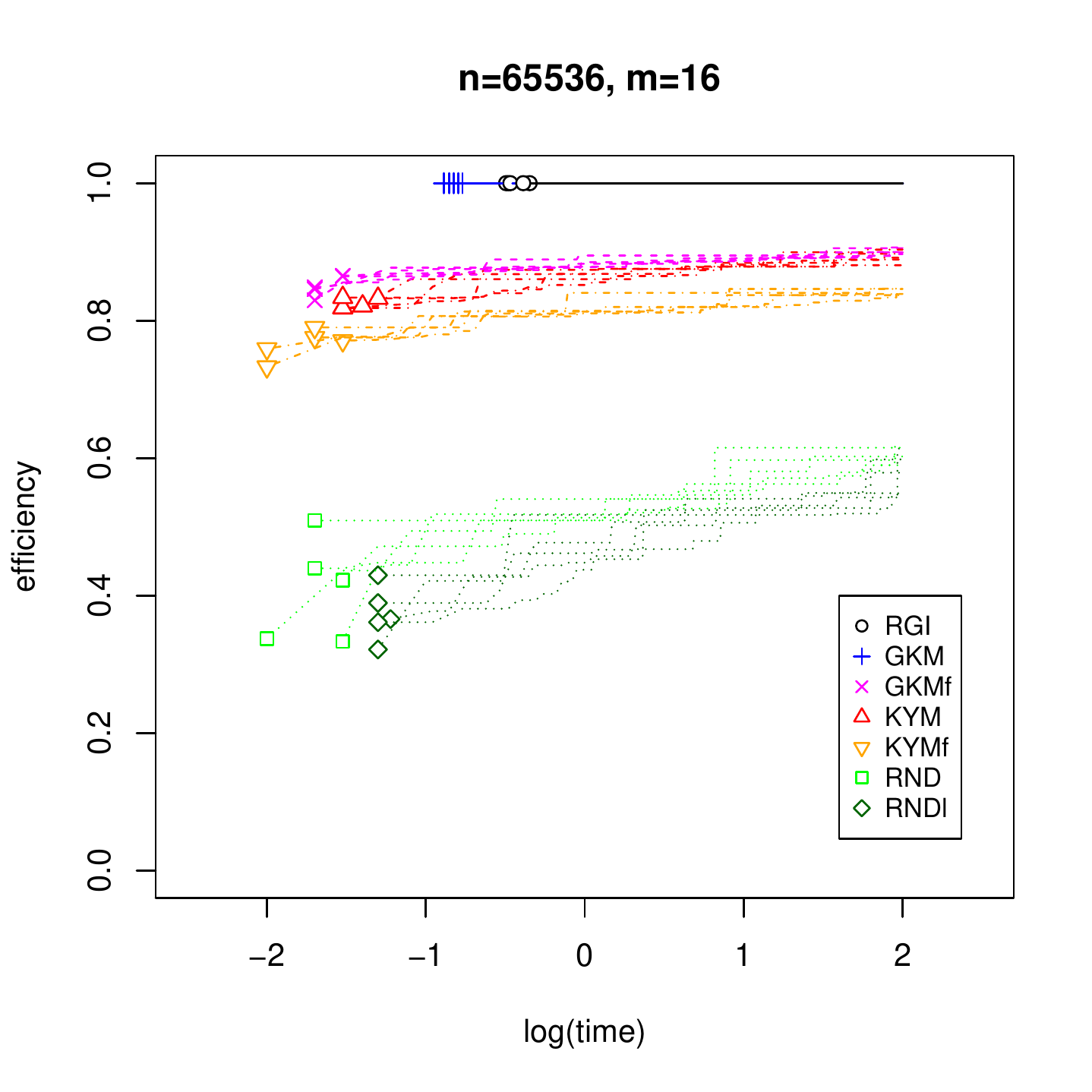}
\caption{The dependence of the best efficiency on the computation time achieved by the multistart approach.  Regressors are deterministically set to $\Fc=\{-1,1\}^m$, $m=16$. Note that for this model the RGH and the GKM produce the perfectly optimal saturated subset $\Ss^*$ such that $\F(\Ss^*)$ is a $16 \times 16$ Hadamard matrix.}\label{ProfileFact}
\end{figure}

\section{Discussion}

We showed that a good approach to the construction of $D$-efficient saturated subsets is to use the efficient version of the Galil-Kiefer method or the modified Kumar-Yildirim method. Nevertheless, the regularized greedy method and the random sampling can also be reasonable choices for some applications (the regularized greedy method if we are more concerned about the efficiency than the speed and the random sampling in the opposite situation). These heuristics can be directly generalized for the construction of non-saturated subsets. For a large size of the initial set, the pre-sampling strategy can be recommended. To minimize the risk of a singular subset, we can use the multi-run approach. Finally, the regularized greedy method and the closely related Galil-Kiefer method can be modified for the use with a different criterion than $D$-optimality.

\section*{Appendix - R Scripts} 

\noindent\textbf{Regularized greedy heuristic}
\begin{verbatim}
function(Fx, del=1e-4) {
# Fx is the nxm matrix of regressors
  m <- ncol(Fx); one <- rep(1, m)
  M <- del * diag(m)
  j <- which.max((Fx^2) %*% one)
  ind <- rep(0, m); ind[1] <- j
  for(i in 2:m) {
   M <- M + tcrossprod(Fx[j, ])
   j <- which.max(((Fx %*%
   t(chol(solve(M))))^2) %*% one)
   ind[i] <- j
  }
  return(ind)
}
\end{verbatim} 

\noindent\textbf{Galil-Kiefer method}
\begin{verbatim}
function(Fx) {
# Fx is the nxm matrix of regressors
  m <- ncol(Fx)
  v2 <- (Fx^2) %*% rep(1, m) 
  j <- which.max(v2)
  ind <- rep(0, m); ind[1] <- j
  for(i in 2:m) {
    scp <- Fx %*% Fx[j, ]
    Fx <- Fx - scp %*% t(Fx[j, ]) / v2[j]
    v2 <- v2 - scp^2 / v2[j]
    j <- which.max(v2); ind[i] <- j
  }
  return(ind)
}
\end{verbatim}

\noindent\textbf{Kumar-Yildirim method}
\begin{verbatim}
function(Fx) {
# Fx is the nxm matrix of regressors
  m <- ncol(Fx); P <- diag(m)
  j <- which.max(abs(Fx %*% rnorm(m)))
  ind <- rep(0, m); ind[1] <- j
  for(i in 2:m) { 
    fx <- P %*% Fx[j,]
    P <- P - tcrossprod(fx) / sum(fx^2)
    j <- which.max(abs(Fx %*%
                      (P %*% rnorm(m))))
    ind[i] <- j
  }
  return(ind)
}
\end{verbatim}


\end{document}